\documentclass[runningheads]{llncs}
\usepackage{graphicx}
\usepackage{xcolor}

\newcommand{\rmv}[1]{} 
\newcommand{\prova}[1]{}

\usepackage{algorithm}
\usepackage{algpseudocode}
\algtext*{EndIf}
\algtext*{EndFor}

\usepackage{amsfonts,amssymb,amscd,euscript,enumerate,verbatim,calc}
\usepackage{amssymb,amsmath}
\usepackage[utf8]{inputenc}

\begin{document}
\title{Structure-aware combinatorial group testing: a new method for pandemic screening
%\thanks{\todo{Supported by organization x.}}
}

\author{Thaís Bardini Idalino\inst{1} \and
Lucia Moura\inst{2}}

\authorrunning{T. B. Idalino and L. Moura}
\titlerunning{Structure-aware combinatorial group testing}
\institute{Universidade Federal de Santa Catarina, Santa Catarina, Brazil
\email{thais.bardini@ufsc.br}\\
\and
University of Ottawa, Ottawa, Canada
\email{lmoura@uottawa.ca}}
\maketitle              % typeset the header of the contribution
\begin{abstract}
Combinatorial group testing (CGT) is used to identify defective items from a set of items by grouping them together and performing a small number of tests on the groups.
Recently, group testing has been used to design efficient COVID-19 testing, so that resources are saved while still identifying all infected individuals. Due to test waiting times, a focus is given to non-adaptive CGT, where groups are designed a priori and all tests can be done in parallel. The design of the groups can be done using Cover-Free Families (CFFs). The main assumption behind CFFs is that a small number $d$ of positives are randomly spread across a  population of $n$ individuals. However, for infectious diseases, it is reasonable to assume that infections show up in clusters of individuals with high contact (children in the same classroom within a school, households within a neighbourhood, students taking the same courses within a university, people seating close to each other in a stadium). The general structure of these communities can be modeled using hypergraphs, where vertices are items to be tested and edges represent clusters containing high contacts. We consider hypergraphs with non-overlapping edges and overlapping edges (first two examples and last two examples, respectively). We give constructions of what we call \emph{structure-aware} CFF, which uses the structure of the underlying hypergraph.  We revisit old CFF constructions, boosting the number of defectives they can identify by taking the hypergraph structure into account. We also provide new constructions based on hypergraph parameters.
\end{abstract}
\section{Introduction}

Group testing literature dates back to the Second World War as an efficient way of testing blood samples for syphilis screening~\cite{DuHwang,Dorfman}. The idea consists of grouping blood samples together before testing, so that negative results could save hundreds of individual tests. This idea was then applied to many other areas: screening vaccines for contamination, building clone libraries for DNA sequences, data forensics for altered documents, modification tolerant digital signatures~\cite{far1997,eps2007,IPL,IWOCA,AMC,PhDThesis,indocrypt,TCS}.
Currently, it is considered a promising scheme for saving time and resources in COVID-19 testing~\cite{TNYT,nature,Nikolopoulos1,Nikolopoulos2,Verdun}. In fact, several countries, such as China, India, Germany and the United States, have adopted group testing as a way of saving time and resources~\cite{nature}. 

In combinatorial group testing (CGT), we are given $n$ items of which at most $d$ are defective (or contaminated). We assume we can test any subset of items, and if the result of the test is positive the subset contains at least one defective (contaminated) item, and if it is negative all items in the subset are non-defective (uncontaminated). The main goal is to minimize the number $t$ of tests for given $n$ and $d$, while determining all defective items.
For a comprehensive treatment, see the text by Du and Hwang~\cite{DuHwang}.

Group testing may be adaptive or non-adaptive~\cite{DuHwang}. \emph{Adaptive} CGT allows us to decide the next tests according to the results of previous tests. This is the case of the binary spliting algorithm, which meets the information theoretical lower bound of  $d \log (n/d)$ tests. In this paper, we focus on \emph{non-adaptive} CGT. Due to test waiting times, non-adaptive CGT  is a useful approach, since we decide all groups at once and can run tests in parallel. In addition, in non-adaptive CGT, we can have more balanced sizes of the groups (items in each test), which is limited in some real applications. For COVID-19 screening, researchers are testing how many samples can be grouped together without compromising the detection of positive results~\cite{nature,Verdun}. 

Items and tests in CGT can be represented by a binary matrix where items correspond to columns and tests correspond to rows, where a 1 means a test uses an item. A $d$-cover free family (or $d$-CFF($t,n$)) is a $t \times n$ matrix with special properties that guarantee the identification of $d$ defective items among $n$ items using $t$ tests and a simple decoding algorithm that takes time $O(tn)$ (see Section~\ref{sec:cff}).

In this paper, we are interested in applications where the defective items are more likely to appear together in predictable subsets of items, which are given as edges of a hypergraph. For example, if we want to monitor a highly transmissible disease among students in a school, classrooms can be the edges (or regions) where it is more likely that if there is one infected individual we may find many. In this way, outbreaks may be detected early while only a few classrooms have infected students.  In this model, we are given a hypergraph where items are vertices and regions are edges such that there are at most $r$ edges that together contain all defective vertices. The objective is still to minimize the number of tests while identifying all defective items. A weaker version of the problem consists of simply identify all infected edges.  In this paper we initiate a more systematic study of how to build CFFs for combinatorial group testing under the hypergraph model, which we call \emph{structure-aware} cover-free families.

{\em Recent related work.}
A similar hypergraph model has been recently proposed as group testing in connected and overlapping communities in the context of COVID-19 testing~\cite{Nikolopoulos1,Nikolopoulos2} and as variable cover-free families motivated by problems in cryptography~\cite{PhDThesis}. 
The work in~\cite{Nikolopoulos1,Nikolopoulos2} span both adaptive and nonadaptive CGT algorithms, but there is not much emphasis on CGT matrix contructions. Our work is on efficient cover-free family constructions for the hypergraph model.
The idea of structure-aware CFF was introduced in the first author's PhD thesis~\cite{PhDThesis} under the name of \emph{variable} CFFs (VCFFs) with an equivalent definition. This was inspired by applications in cryptography, where they would allow for location of clustered modifications in a signed document when using modification-tolerant digital signatures.

{\em Our results and paper structure.}
Basic concepts for cover-free families are given in Section~\ref{sec:cff}. The new definitions of structure-aware cover free families and edge-identifying CFFs are given in Section~\ref{sec:VCFF} along with related decoding algorithms. 
CFF constructions for hypegraphs with non-overlaping edges are given in Section~\ref{sec:disjoint}. In particular, we revisit known $d$-CFF constructions (Sperner, product, array group testing, %2-designs,
polynomials in finite fields) and show how they can be viewed as a \emph{structure-aware} CFF, allowing a much larger defect identification when items are clustered into conveniently chosen hypergraphs. We exemplify how these hypergraphs relate to realistic community-like structures. In a generalization of the Sperner construction ($r=1$) we also give results under the more realistic assumption of limited number of samples per tests (Section~\ref{sec:sperner}). CFF constructions for the more general case of hypergraphs with overlapping edges are given in Section~\ref{sec:overlap}. We give constructions for both $r=1$ and $r>1$ using edge-colouring and strong edge-colouring of hypergraphs, to partition the  hypergraph into non-overlapping subgraphs that can be constructed using results from the previous section. Proofs are 
in the appendix for refereeing purposes.

\section{Cover-Free Families} \label{sec:cff}

Cover-free families were first introduced by Kautz and Singleton~\cite{Kautz} in the context of \emph{superimposed codes}. 
They are found under different names, such as $d$-disjunct matrices and strongly selective families \cite{DuHwang,PR}. We can define $d$-CFF via a matrix or a set system.

\begin{definition}[CFF via matrix]\label{defCFFmatrix} Let $d$ be a positive integer. A $d$-cover-free family, denoted $d$-CFF$(t,n)$, is a $t \times n$ 0-1 matrix where the submatrix given by any set of $d+1$ columns contains a permutation matrix (each row of an identity of order $d+1$) among its rows.
\end{definition}

A set system $\mathcal{F} = (X, \mathcal{B})$ consists of a set $X$ and a collection $\mathcal{B}$ of subsets of $X$. The {\em set system associate to matrix $\mathcal{M}$} is the set system $\mathcal{F}_{\mathcal{M}} = (X, \mathcal{B})$ with $X$ corresponding to rows and $\mathcal{B}$  corresponding to columns of $\mathcal{M}$, where $B_i\subseteq \mathcal{B}$ has column $i$ as its characteristic vector, $1\leq i \leq n$.
A $d$-CFF can be equivalently defined in terms of its set system $\mathcal{F}_{\mathcal{M}}$, by specifying that no set of $d$ columns ``covers'' any other column.
 
\begin{definition}[CFF via set system]\label{defCFF} Let $d$ be a positive integer. A $d$-cover-free family, denoted $d$-CFF$(t,n)$, is a set system $\mathcal{F} = (X, \mathcal{B})$ with $|X| = t$ and $|\mathcal{B}| = n$ such that for any $d+1$ subsets $B_{i_0}, B_{i_1}, \ldots, B_{i_d} \in \mathcal{B}$, we have
	\begin{equation}\label{property:cff}
	|B_{i_0} \setminus \bigcup_{j=1}^{d}B_{i_j}| \geq 1.
	\end{equation}	
\end{definition}

Next we show an example of a $2$-CFF($9,12$), which can be used to test $n=12$ items with $t = 9$ tests and identify up to $d=2$ defective items.

\begin{center}
  {\tiny	\centering	\setcounter{MaxMatrixCols}{12}
	$\mathcal{M}= \begin{pmatrix}
	1&  0 & 0 & 1 & 0 & 0&  1 & 0 & 0 & 1 & 0 & 0\\
	1 & 0 & 0 & 0 & 1 & 0 & 0 & 1 & 0 & 0 & 1 & 0 \\
	1 & 0 & 0 & 0 & 0 & 1 & 0 & 0 & 1 & 0 & 0 & 1 \\
	0 & 1 & 0 & 1 & 0 & 0 & 0 & 0 & 1 & 0 & 1 & 0 \\
	0 & 1 & 0 & 0 & 1 & 0 & 1 & 0 & 0 & 0 & 0 & 1 \\
	0 & 1 & 0 & 0 & 0 & 1 & 0 & 1 & 0 & 1 & 0 & 0 \\
	0 & 0 & 1 & 1 & 0 & 0 & 0 & 1 & 0 & 0 & 0 & 1 \\
	0 & 0 & 1 & 0 & 1 & 0 & 0 & 0 & 1&  1 & 0 & 0 \\
	0 & 0 & 1 & 0 & 0 & 1 & 1 & 0 & 0 & 0 & 1 & 0  \\
	\end{pmatrix}$}
	{\small 
	$\begin{array}{l}
	X = \{1, 2, \ldots, 9\}\\
	B_1 = \{1,2,3\}, B_2 = \{4,5,6\}, \ldots, B_{12} = \{3,5,7\}\\
	\mathcal{B} = \{B_1, B_2, \ldots, B_{12}\}
	\end{array}$}
\end{center}

After running the tests on groups of items according to the rows of a $d$-CFF matrix $\mathcal{M}$, we can run a simple algorithm to identify the invalid items. 
When we apply Algorithm~\ref{alg:general} with a $d$-CFF matrix $\mathcal{M}$ and the number of defectives is indeed bounded by $d$, then after the first loop $x$ has at most $d$ nonzero components. 
So for $d$-CFF, the second loop can be removed and substituted by  
a simple check that the number of 1's in $x$ does not exceed $d$; in this case, the output will be  Boolean, i.e. every component is in $\{0,1\}$, and correct. We give this more general algorithm, used in Section~\ref{sec:VCFF}.
In the case of other types of matrices or when the hypothesis of testing are not satisfied, the algorithm classifies the items into three types of defective status (yes, no, maybe) according to the information provided by test results. Assuming correct outcome for group testing, the items with $x_i\in\{0,1\}$ do not give false positive/negative results.

\vspace{-10pt}
\begin{algorithm}[h]
\caption{Non-adaptive CGT algorithm to identify invalid items}\label{alg:general}
\begin{algorithmic}
\item \textbf{Input:} Group testing matrix $\mathcal{M}$ and test result $y = (y_1, y_2, \ldots, y_t)$, with $y_i = 1$ iff the i-th test was positive.
\item \textbf{Output:} Vector $x = (x_1, x_2, \ldots, x_n)$, $x_j = 1,0, 0.5$ if the j-th item is defective,  nondefective, unknown, respectively.
\item $x\gets (1, ..., 1)$
\For{i = 1, \ldots, t}
  \For{j = 1, \ldots, n}
       \If{$\mathcal{M}_{i,j} = 1$ and $y_i = 0$}
        %\State 
        $x_j \gets 0$
        \EndIf
    \EndFor
\EndFor

\For{j such that $x_j=1$} 

 \If {$\exists i$ such that ($\mathcal{M}_{i,j}=1$ and ($x_\ell=0$, $\forall \ell\not=j$ with $\mathcal{M}_{i,\ell}=1$))}
 \State{$x_j \gets 1$} \Comment{Item $j$ is on a failing test together with only non-defective items}
 \Else  \ $x_j \gets 0.5$ \Comment{Can't guarantee $j$ is the cause of failures but maybe defective}
   \EndIf
 \EndFor
 
\State \Return $x$
\end{algorithmic}
\end{algorithm}
\vspace{-10pt}

For a given $n$ and $d$, we are interested in constructing $d$-CFFs with the smallest possible $t$, so we define $t(d,n) = \text{min}\{t: \exists \ d\text{-CFF}(n,t)\}$. For $d = 1$, Sperner's theorem gives an optimal construction for $1$-CFFs. The value $t$ grows as $\log_2 n$ as $n \rightarrow \infty$, which meets the information theoretical lower bound on the number of bits necessary to uniquely distinguish the $n$ inputs. For $d \geq 2$, the best known lower bound on $t$ for $d$-CFF($t,n$) is given by $t(d,n) \geq c\frac{d^2}{\log d } \log n$ for some constant $c$~\cite{furedi,Ruszinko,wei2006}, with $c$ proven to be $\approx 1/4$ in~\cite{furedi} and $\approx 1/8$ in~\cite{Ruszinko}.

For $d\geq 2$, there are several approaches to construct $d$-CFFs, for example, we can use codes and combinatorial designs~\cite{monotone}. Probabilistic methods usually provide the best existence results known, and derandomization techniques can be used to yield efficient algorithms to construct CFFs, such as in~\cite{Bshouty,PR,Vaccaro2,Vaccaro}. 
Using this approach, polynomial time algorithms exist to construct a $d$-CFF($n,t$) with $t = \Theta(d^2 \log n)$~\cite{Bshouty,PR,Vaccaro2,Vaccaro}.

\section{Structure-aware Cover-Free Families} \label{sec:VCFF}
In this section, we define structure-aware cover-free families (SCFFs) by adding a hypergraph structure to a CFF. Vertices correspond to columns and edges specify sets of columns where defective items may appear more likely together. We use the assumption that defective items are contained in a small number $r$ of edges inside of which any number of defective items may be found.  
For example, the outbreak of a disease in a school/university could be detected by associating vertices with students, edges with classrooms/courses; even if the number of infected students is high, the CFF would detect them as long as they are concentrated in a small number of classrooms/courses.

\begin{definition}[Structure-aware CFFs]\label{variabledef}
    Let $n,t> 0$ and $r\geq0$ be integers.
    Let $\mathcal{H} = ([1,n], \mathcal{S})$ be a hypergraph with $n$ vertices and $m$ edges, and let $\mathcal{M}$ be a $t \times n$ binary matrix with associated set system $\mathcal{F}_{\mathcal{M}} = ([1,n], \mathcal{B})$, $\mathcal{B} = \{B_1, \ldots, B_n\}$.  Matrix $\mathcal{M}$ is a {\em structure-aware cover-free family}, denoted $(\mathcal{S},r)$-CFF($t,n$), if for any $r$-set of hyperedges $\{S_1,\ldots,S_r\}\subseteq \mathcal{S}$, and for any $I \subseteq \cup_{j=1}^r S_j$ and any $i_0\in [1,n]\setminus I$, we have
	\begin{equation}
	\bigg|B_{i_0} \Big\backslash \bigg(\bigcup_{i \in I}B_{i}\bigg)\bigg| \geq 1. \label{eqn:vcff}
	\end{equation}	
\end{definition}

We observe that a $d$-CFF($t,n$) is equivalent to an $(\mathcal{S},d)$-CFF($t,n$) where edges are singleton vertices $\mathcal{S}=\{\{1\},\{2\},\ldots \{n\}\}$.

We now consider how the status of edges influence the detectability of defective items.
An edge is {\em defective} if it contains a defective vertex and {\em non-defective}, otherwise. A set of edges is a {\em defect cover} if the set of defective vertices is contained in the union of these edges; such a defective cover is {\em minimal} if no proper subset is a defect cover. A minimal defect cover is always contained in the set of defective edges, but the number of defective edges may be much larger 
than the size of a defect cover for hypergraphs with overlaping edges. The next proposition shows that a structure-aware CFF ability to detect defectives only depends on the cardinality of a minimum defect cover being bounded by $r$.

\begin{proposition}
Let $\mathcal{H} = ([1,n], \mathcal{S})$ be a hypergraph, $\cal{M}$ be an $(\mathcal{S},r)$-CFF($t,n$) and $y\in\{0,1\}^t$ be the result of tests given by $\cal{M}$ on items $1,\ldots, n$. 
If $\mathcal{H}$ has a defect cover with at most $r$ edges then Algorithm~\ref{alg:general} on inputs $(\mathcal{M},y)$ returns a Boolean output $x$ such that $x_i=1$ if and only if item $i$ is defective.
\end{proposition}
\begin{proof}
Let $\mathcal{DC}=\{e_1,e_2,\ldots,e_{\ell}\}$ be a defect cover with $\ell\leq r$.
Let ${i_0}\in[1,n]$ be an item and
take  $I=(\cup_{i=1}^{\ell} e_i)\setminus \{{{i_0}}\}$.
Since $\mathcal{M}$ is a $(\mathcal{S},r)$-CFF($t,n$), Equation (\ref{eqn:vcff}) guarantees there exists a row $w$ in $\mathcal{M}$ that tests item ${i_0}$ and avoids all other defective items.
If item ${{i_0}}$ is non-defective, this row will be a passing test, $y_w=0$, and $x_{i_0}$ will be set to $0$ in the first loop.
Otherwise, item ${i_0}$ is defective, and $x_{i_0}$ will remain equal to  $1$ at the end of the first loop.
In addition, row $w$ will prove that the condition on the second loop is false for $i=i_0$ so $x_{i_0}$ will never be set to 0.5.
Therefore, the output will be a Boolean $x$ that correctly informs the status of the items. \qed
\end{proof}

We are also interested in identifying infected edges when the output of Algorithm~\ref{alg:general} is not Boolean, which can happen if defective items are spread over too many edges (defective covers have size $>r$).  For example, in schools the tests may not provide full information on infected students, but we still may extract information on which classrooms are infected.
The following algorithm provides edge information based on ternary vertex information for a hypergraph $\cal{H}$.

\begin{algorithm}[h]
\caption{Edge information from vertices}\label{alg:edgeinfo}
\begin{algorithmic}
\item \textbf{Input:} Hypergraph $\mathcal{H}=(V,E)$ with $n$ vertices and $m$ edges;
Group testing matrix $\cal{M}$, boolean results $y=(y_1,y_2,\ldots,y_t)$;
Vector $x = (x_1, x_2, \ldots, x_n)$, $x_j = 1,0, 0.5$ if the j-th item is defective,  non-defective, unknown, respectively.
\item \textbf{Output:} Vector $z = (z_1, z_2, \ldots, z_m)$, $z_e = 1,0, 0.5$ if the e-th edge is defective,  nondefective, unknown, respectively.

\For{$s = 1, \ldots, m$} \Comment{this loop gets edge status from vertices}
    \State $z_s\gets  0$; 
    \For{each vertex $v_i$ in edge $e_s$}
        \If{$x_i=1$}
          \ $z_s\gets 1$
        \Else
          \If{$x_i=0.5$ and $z_s=0$}
            \ $z_s\gets 0.5$
          \EndIf
        \EndIf
    \EndFor
\EndFor
\For{$i=1,\ldots,t$} \Comment{this loop gets edge status from test results}
   \If{$y_i=1$} 
       \State{$E=\{j: M_{i,j}=1\ \mathrm{and}\ x_j \not=0\}$ }
       \For{$s = 1, \ldots, m$}
          \If{($z_s=0.5$) and ($E\subseteq e_s$)} \ $z_s\gets 1$
          \EndIf
        \EndFor  
   \EndIf
\EndFor
\State \Return $z$
\end{algorithmic}
\end{algorithm}

Some CFFs may have a value of $r$ for vertex status identification but have a larger value $r$ for edge status identification. This can be useful for applications, in that infected communities are identifiable even though we do not have perfect individual identification. To capture this property, we define edge-identifying CFFS (ECFFs), which has a weaker coverage requirement than SCFFs.

\begin{definition}[Edge-identifying CFFs]\label{edgeVCFF}
Let $r$, $t$, $n$, $\mathcal{M}$, $\mathcal{H}$ and $\mathcal{F}_{\mathcal{M}}$ be as in Definition~\ref{variabledef}.
 We say $\mathcal{M}$ is an $(\mathcal{S},r)$-ECFF($t,n$) if for any $\ell$-subset of hyperedges $\{S_1,\ldots,S_\ell\} \subseteq \mathcal{S}$, $\ell\leq r$,  and any ${i_0}\notin S=\cup_{j=1}^{\ell} S_j$, we have
	\begin{equation} \label{eq:ecff}
	\bigg|B_{i_0} \Big\backslash \bigg(\bigcup_{i \in S}B_{i}\bigg)\bigg| \geq 1. 
	\end{equation}	
\end{definition}

\begin{proposition}
Let $\mathcal{H} = ([1,n], \mathcal{S})$ be a hypergraph, $\mathcal{M}$ be a $(\mathcal{S},r)$-ECFF($t,n$), and $y$ be the test results for $\mathcal{M}$.
Let $x$ be the output of Algorithm~\ref{alg:general} for inputs $(\mathcal{H},\mathcal{M},$y$)$.
Then, if $\mathcal{H}$ has a defect cover with at most $r$ edges then Algorithm~\ref{alg:edgeinfo} applied to $(\mathcal{H}, \mathcal{M}, $x$, $y$)$ returns an output $z$  such that $\{S_j\in \mathcal{S}: z_j=1\}$ forms a defect cover.
\end{proposition} 
\begin{proof} Let $\mathcal{IC}=\{S_{e_1},\ldots,S_{e_{\ell}}\}$ be any minimal defect cover with $\ell\leq r$ and let $C=\cup_{i=1}^{\ell} S_{e_i}$.
Then, any item ${i_0}\not\in C$ is non-defective and Equation (\ref{eq:ecff}) guarantees there is a row $w$ that tests ${i_0}$ and avoids all items in $C$, and thus avoids all defective items, which means $y_w=0$ and Algorithm 1 sets $x_{{i_0}}=0$.
Now, consider any edge $S_e \in \mathcal{IC}$ and let $S=U_{X\in \mathcal{IC}\setminus\{S_e\}} X$. Since $\mathcal{IC}$ is minimal, $S_e$ must contain a defective item $u \in S_e\setminus S$.
By Equation (\ref{eq:ecff}), using $i_0=u$, there must be a test/row $w$ that contains $u$ and avoids $S$.  
Thus, we must have $\{j: M_{w,j}=1\ \mathrm{and}\ x_j \not=0\} \subseteq S_e$, which implies Algorithm~\ref{alg:edgeinfo} sets $z_e=1$.
Therefore, $z_{j} =1$ for all ${S_j}\in \mathcal{IC}$ and possibly for a few other edges. Since every superset of an defect cover is a defect cover 
$\{S_j\in \mathcal{S}: z_j=1\}$ is a defect cover.
\qed
\end{proof}

For any CFF, structure-aware CFF, or ECFF matrix $\mathcal{M}$ we denote by $L_{\mathcal{M}}$
the number of ones in each row  of ${\mathcal M}$.
We keep track of these quantity in some constructions, since we may have limit $L_{max}$ on the number of ones per row, in cases where combining too many samples can result on a false negative.

\section{Structure-aware CFFs: non-overlapping edges}
\label{sec:disjoint}

We revisit old CFF constructions and show we can boost the number of defectives it can identify by taking a suitable hypergraph structure into account. We also propose some new constructions. Here we consider the case of non-overlapping edges, meaning that items do not participate in more than one edge. 

\subsection{Sperner-type constructions for $r=1$}\label{sec:sperner}

A Sperner set system is a set system where no set is contained in any other set in the set system. Sperner's theorem states that the largest Sperner set system on an $t$-set is formed by taking all subsets of cardinality $\lfloor t/2\rfloor$.
Given $n$, a $1$-CFF$(t,n)$ with minimum $t$ is obtained from Sperner  theorem by taking $t=\min\{s: {s \choose {\lfloor s/2\rfloor}} \leq n\}$ and the corresponding matrix having the characteristic vectors of $\lfloor t/2\rfloor$-subsets as columns.
We note that $t\sim \log n$ and this is the best possible, since being 1-CFF is equivalent to being Sperner.

A Sperner set system with sets with cardinality $a<t/2$ can be used as a $1$-CFF if
${t-1 \choose {\lfloor t/2\rfloor}-1}$ exceeds a maximum allowed number of ones per row, $L_{max}$. For nonoverlapping hypergraphs and $r=1$, we give constructions for SCFF for both unlimited and limited $L_{\mathcal{M}}$.

\begin{proposition}[$r=1$, unlimited $L_{\mathcal{M}}$]\label{prop:spernerunlimited}
Let $\mathcal{H}=([1,n],\mathcal{S})$ be a hypergraph with $m$ disjoint edges of cardinality at most $d$ that span $[1,n]$.
Let $\mathcal{M}$ be the vertical concatenation of matrices $M_1$ and $M_2$.
Let $M_1$ be obtained from a $1$-CFF$(t_1,m)$ matrix $A$ with  $t_1=\min\{s: {s \choose {\lfloor s/2\rfloor}} \leq m\}$ in such a way that if vertex $v_i$ is incident to edge $b_j$ column $i$ of $M_1$ repeats column $j$ of $A$. Let $M_2$ be a $d \times n$ matrix with an identity matrix of dimension up to $|S|$ pasted under the items of each edge $S\in\mathcal{S}$.
 Then, $M_1$ is an $(\mathcal{S},1)$-ECFF$(t_1,n)$ and $\mathcal{M}$ is an $(\mathcal{S},1)$-CFF$(t_1+d,n)$.
\end{proposition}

For uniform hypergraphs the construction above gives $t\sim \log m + d = \log n/d +d$, but does not limit $L_{\mathcal{M}}$. The next proposition is useful for limited $L_{\mathcal{M}}$, as shown in the example that follows it.

\begin{proposition}[$r=1$, $L_{\mathcal{M}}\leq L_{max}$]\label{prop:spernerlimited}
Let $L_{max}$ be a positive integer that limits the number of 1s in each row of the CFF.
Let $\mathcal{H}=([1,n],\mathcal{S})$ be a hypergraph with $m$ disjoint edges of cardinality at most $d$ that span $[1,n]$, where $d\leq L_{max}$.
Let $t_1=\min\{s: {s \choose {\lfloor s/2\rfloor}} \leq m\}$. Then,
\begin{enumerate}
\item If $ d \times{t_1-1 \choose \lfloor t_1/2\rfloor -1}\leq L_{max}$ and $m \leq L_{max}$ then 
$\mathcal{M}$ given in Proposition~\ref{prop:spernerunlimited} is an $(\mathcal{S},1)$-CFF$(t_1+d,n)$ with $L_{\mathcal{M}}\leq L_{max}$.
\item Otherwise, let $q=\lceil m/L_{max} \rceil$. 
Take $t,a$ such that ${t \choose a} \geq m$ and $d\times {t-1 \choose \lfloor a\rfloor -1} \leq L_{max}$.
Then, there exists a $(\mathcal{S},1)$-CFF$(t+qd,n)$ matrix $\mathcal{M}$ with $L_{\mathcal{M}}\leq L_{max}$.
\end{enumerate}

\end{proposition}
\begin{proof}
The first statement comes from Proposition~\ref{prop:spernerunlimited}. The second statement comes from vertically concatenating $N_1$ and $N_2$ where $N_1$ is formed by a Sperner system of
$a$-subsets of a $t$-set on the $m$ edges (repeating columns for vertices in the same edges) and $N_2$ is build similarly to $M_2$ in Proposition~\ref{prop:spernerunlimited} but splitting rows (the ones in each row are split into up to $q$ new rows not exceeding $L_{max}$). \qed
\end{proof}

\begin{example}[Proposition~\ref{prop:spernerlimited} used for $m$ classrooms with $d$ students each]
Suppose $n$ students are divided into $m$ classrooms of size up to $d$.
Then Proposition~\ref{prop:spernerlimited} can be used to identify all infected students, provided they are all in a single classroom ($r=1$).
The table below reports on number of tests for each scenario depending on value of $L_{max}$ for the construction on Proposition~\ref{prop:spernerlimited}.
The line with $L=\infty$ shows the number of tests for the construction for unlimited $L$ (Proposition~\ref{prop:spernerunlimited}).
The last line shows the lower bound given in~\cite{furedi} for the number of rows $t$ on a $d$-CFF$(t,n)$ required for location of any set of $d$ infected students, not necessarily concentrated on a single classroom.

\ \\
\tiny
\begin{tabular}{c||c|c|c|c||c|c|c|c||c|c|c|c||c|c|c|c||c|c|c|c||c|c|c|c||c|c|c|c}
$n/100$ 
     &  .5 & 1 &  2 &  3 &  1 & 2  &  4 & 6  & 1.5& 3 &  6  & 9  & 2 & 4   &  8 & 12 & 2.5& 5  & 10 & 15 & 3 & 6  & 12 & 18\\
$m$ &   10& 10 & 10 & 10 & 20 & 20 & 20 & 20 & 30 & 30 & 30 & 30 & 40 & 40 & 40 & 40 & 50 & 50 & 50 & 50 & 60 & 60 & 60 & 60\\
$d$  &  5 & 10 & 20 & 30 & 5  & 10 & 20 & 30 & 5  & 10 & 20 & 30 & 5  & 10 & 20 & 30 & 5  & 10 & 20 & 30 & 5  & 10 & 20 & 30 \\ \hline
$L=$ \\ \cline{2-25}
10 & 11 & 16 & 26 & 36 & 17 & 27 & 47 & 67 & 23 & 38 & 68 & 98 & 28 & 48 & 88 & 128& 33 & 58 &108 &158 & 39 & 69 &129 &189\\
15 & 11 & 16 & 26 & 36 & 17 & 27 & 47 & 67 & 18 & 28 & 48 & 68 & 23 & 38 & 68 & 98 & 28 & 48 & 88 &128 & 29 & 49 & 89 &129\\
20 & 11 & 16 & 26 & 36 & 12 & 17 & 27 & 37 & 18 & 28 & 48 & 68 & 18 & 28 & 48 & 68 & 23 & 38 & 68 & 98 & 24 & 39 & 69 &99\\
25 & 10 & 16 & 26 & 36 & 12 & 17 & 27 & 37 & 18 & 28 & 48 & 68 & 18 & 28 & 48 & 68 & 18 & 28 & 48 & 68 & 24 & 39 & 69 &99\\
30 & 10 & 16 & 26 & 36 & 12 & 17 & 27 & 37 & 14 & 18 & 28 & 38 & 18 & 28 & 48 & 68 & 18 & 28 & 48 & 68 & 19 & 29 & 49 &69\\ \hline
$L=\infty$ & 10 & 15 & 25 & 35 & 11 & 16 & 26 & 36 & 12 & 17 & 27 & 37 & 13 & 18 & 28 & 38 & 13 & 18 & 28 & 38 & 13 & 18 & 28 &38\\ \hline
{\tiny $t(d,n)>$} & 21 & 66 & 180 & 270 & 21 & 66 & 231 & 496 & 21 & 66 & 231 & 496 & 21 & 66 & 231 & 496 &  23& 66 & 231 & 496 & 25 & 66 & 231 & 496\\
\hline
\end{tabular}
\end{example}

\subsection{Kronecker product constructions (general $r$)}

Let $A_k$ be an $m_k \times n_k$ binary matrix, for $k=1,2$, and \textbf{0} be the matrix of all zeroes with same dimension as $A_2$. The Kronecker product $P = A_1 \otimes A_2$ is a binary matrix formed of blocks $P_{i,j}$ such that $P_{i,j} = A_2$ if $A_{1_{i,j}} = 1$ and $P_{i,j} = \textbf{0}$, otherwise. 
We denote by $R_k$ the row matrix with $k$ ones and by $I_k$ the identity matrix of dimension $k$.
The propositions given after each theorem specializes the theorem construction and generalizes to SCFF, boosting the defective detection.

\begin{theorem}[Li et al.~\cite{monotone} for $d=2$, Idalino and Moura~\cite{TCS}]\label{dcff}
	Let $A_1$ be a $d$-CFF$(t_1, n_1)$ and $A_2$ be a $d$-CFF$(t_2, n_2)$, then $C = A_1 \otimes A_2$ is a $d$-CFF$(t_1t_2, n_1 n_2)$.
\end{theorem}

\begin{proposition} \label{prop:kronecker}
Let ${\cal H}=([1,n],\mathcal{S})$ be a hypergraph formed by $m$ disjoint edges of cardinality $k$, $n=k\times m$.
Let $r$ be a positive integer, and let $A$ be an $r$-CFF$(t,m)$.
Then $A\otimes R_k$ is an $(\mathcal{S},r)$-ECFF$(t,km)$ and $A\otimes I_k$ is an $(\mathcal{S},r)$-CFF$(kt,km)$.
\end{proposition}

\begin{theorem}[Li et al.~\cite{monotone} for $d=2$, Idalino and Moura~\cite{TCS}]\label{d-sum}
	Let $d \geq 2$, $A_1$ be a $d$-CFF$(t_1, n_1)$, $A_2$ be a $d$-CFF$(t_2, n_2)$, $B$ be a $(d-1)$-CFF$(s, n_2)$. 
	Let $C$ be the vertical concatenation of $B\otimes A_1$ with $A_2 \otimes R_{n_1}$.
	Then $C$ is a $d-$CFF$(st_1 + t_2, n_1 n_2)$.
\end{theorem}

\begin{proposition}\label{prop:kroneckerplus}
Let ${\cal H}=([1,n],\mathcal{S})$ be a hypergraph formed by $m$ disjoint edges of cardinality $k$, $n=k\times m$.
Let $r$ be a positive integer, $A$ be an $r$-CFF$(t_A,m)$, and $B$ be an $(r-1)$-CFF$(t_B,m)$.
Then the vertical concatenation of $A\otimes R_k$ with $B\otimes I_k$  is an $(\mathcal{S},r)$-CFF$(t_A+k t_B,km)$.
Moreover, if edges have different cardinalities bounded by $k$, a similar construction yields 
an $(\mathcal{S},r)$-CFF$(t_A+k t_B,n)$.
\end{proposition}

\begin{figure}[h]
Construction in Proposition~\ref{prop:kronecker}, using $2$-CFF(9,12) $A$:\\

{\tiny
	\centering	\setcounter{MaxMatrixCols}{12}
	$A= \begin{pmatrix}
	1&  0 & 0 & 1 & 0 & 0&  1 & 0 & 0 & 1 & 0 & 0\\
	1 & 0 & 0 & 0 & 1 & 0 & 0 & 1 & 0 & 0 & 1 & 0 \\
	1 & 0 & 0 & 0 & 0 & 1 & 0 & 0 & 1 & 0 & 0 & 1 \\
	0 & 1 & 0 & 1 & 0 & 0 & 0 & 0 & 1 & 0 & 1 & 0 \\
	0 & 1 & 0 & 0 & 1 & 0 & 1 & 0 & 0 & 0 & 0 & 1 \\
	0 & 1 & 0 & 0 & 0 & 1 & 0 & 1 & 0 & 1 & 0 & 0 \\
	0 & 0 & 1 & 1 & 0 & 0 & 0 & 1 & 0 & 0 & 0 & 1 \\
	0 & 0 & 1 & 0 & 1 & 0 & 0 & 0 & 1&  1 & 0 & 0 \\
	0 & 0 & 1 & 0 & 0 & 1 & 1 & 0 & 0 & 0 & 1 & 0  \\
	\end{pmatrix}$,
	$I_3= \begin{pmatrix} 1 & 0 & 0\\ 0 & 1 & 0\\ 0 & 0 & 1 \end{pmatrix}$,
	\setcounter{MaxMatrixCols}{12}
	$A \otimes I_3 = \begin{pmatrix}
    I_3 & 0 & 0 & I_3 & 0 & 0&  I_3 & 0 & 0 & I_3 & 0 & 0\\
	I_3 & 0 & 0 & 0 & I_3 & 0 & 0 & I_3 & 0 & 0 & I_3 & 0 \\
	I_3 & 0 & 0 & 0 & 0 & I_3 & 0 & 0 & I_3 & 0 & 0 & I_3 \\
	0 & I_3 & 0 & I_3 & 0 & 0 & 0 & 0 & I_3 & 0 & I_3 & 0 \\
	0 & I_3 & 0 & 0 & I_3 & 0 & I_3 & 0 & 0 & 0 & 0 & I_3 \\
	0 & I_3 & 0 & 0 & 0 & I_3 & 0 & I_3 & 0 & I_3 & 0 & 0 \\
	0 & 0 & I_3 & I_3 & 0 & 0 & 0 & I_3 & 0 & 0 & 0 & I_3 \\
	0 & 0 & I_3 & 0 & I_3 & 0 & 0 & 0 & I_3 & I_3 & 0 & 0 \\
	0 & 0 & I_3 & 0 & 0 & I_3 & I_3 & 0 & 0 & 0 & I_3 & 0  \\
	\end{pmatrix}$.
}
\ \\

Construction in Proposition~\ref{prop:kroneckerplus}, using $2$-CFF(9,12) $A$ and $1$-CFF(6,12) $B$:\\

{\tiny
	\centering	\setcounter{MaxMatrixCols}{12}
	$B= \begin{pmatrix}
	1&  1 & 1 & 1 & 1 & 1&  0 & 0 & 0 & 0 & 0 & 0\\
	1 & 1 & 1 & 1 & 0 & 0 & 1 & 1 & 0 & 0 & 0 & 0 \\
	1 & 0 & 0 & 0 & 1 & 1 & 0 & 0 & 1 & 1 & 1 & 0 \\
	0 & 1 & 0 & 0 & 1 & 0 & 1 & 0 & 1 & 1 & 0 & 1 \\
	0 & 0 & 1 & 0 & 0 & 1 & 0 & 1 & 1 & 0 & 1 & 1 \\
	0 & 0 & 0 & 1 & 0 & 0 & 1 & 1 & 0 & 1 & 1 & 1 \\
	\end{pmatrix}$,
	$R_3= \begin{pmatrix} 1 & 1 & 1\\ \end{pmatrix}$,
	\setcounter{MaxMatrixCols}{1}
	$\begin{pmatrix}A  \otimes R_3\\ \hline B  \otimes I_3\end{pmatrix}$ 
	\setcounter{MaxMatrixCols}{12}$= \begin{pmatrix}
	R_3 & 0 & 0 & R_3 & 0 & 0&  R_3 & 0 & 0 & R_3 & 0 & 0\\
	R_3 & 0 & 0 & 0 & R_3 & 0 & 0 & R_3 & 0 & 0 & R_3 & 0 \\
	R_3 & 0 & 0 & 0 & 0 & R_3 & 0 & 0 & R_3 & 0 & 0 & R_3 \\
	0 & R_3 & 0 & R_3 & 0 & 0 & 0 & 0 & R_3 & 0 & R_3 & 0 \\
	0 & R_3 & 0 & 0 & R_3 & 0 & R_3 & 0 & 0 & 0 & 0 & R_3 \\
	0 & R_3 & 0 & 0 & 0 & R_3 & 0 & R_3 & 0 & R_3 & 0 & 0 \\
	0 & 0 & R_3 & R_3 & 0 & 0 & 0 & R_3 & 0 & 0 & 0 & R_3 \\
	0 & 0 & R_3 & 0 & R_3 & 0 & 0 & 0 & R_3 &  R_3 & 0 & 0 \\
	0 & 0 & R_3 & 0 & 0 & R_3 & R_3 & 0 & 0 & 0 & R_3 & 0  \\ \hline
	I_3&  I_3 & I_3 & I_3 & I_3 & I_3&  0 & 0 & 0 & 0 & 0 & 0\\
	I_3 & I_3 & I_3 & I_3 & 0 & 0 & I_3 & I_3 & 0 & 0 & 0 & 0 \\
	I_3 & 0 & 0 & 0 & I_3 & I_3 & 0 & 0 & I_3 & I_3 & I_3 & 0 \\
	0 & I_3 & 0 & 0 & I_3 & 0 & I_3 & 0 & I_3 & I_3  & 0 & I_3  \\
	0 & 0 & I_3  & 0 & 0 & I_3 & 0 & I_3  & I_3  & 0 & I_3  & I_3  \\
	0 & 0 & 0 & I_3 & 0 & 0 & I_3 & I_3 & 0 & I_3 & I_3 & I_3 \\
	\end{pmatrix}$.
}
	\caption{Two $(\mathcal{S},2)$-CFF($27,36$), $\mathcal{S}$ consists of 12 disjoint edges of size 3. Up to six defective items concentrated within 2 edges can be identified.}
	\label{variable-cff-example}
\end{figure}
\vspace{-5mm}

\subsection{Array and Hypercube Constructions} \label{hyperpoly}
An array-based scheme for group testing uses an $n_1 \times n_2$ array, where each entry of the array corresponds to an item to be tested and the tests are performed on rows and columns, for a total of $n_1+n_2$ tests. This can be used on a 2-stage algorithm, where all items at the intersection of a positive row and column should be individually tested in a second stage to solve ambiguities~\cite{Sudbury,Hudgens,Kim}. {For $d=1$ defective item, one stage is enough.} {Figure~\ref{array} (a) shows a $5\times 5$ array with defective items in red}. 
This idea can be generalized to higher dimensions, constructing an $n_1\times\ldots \times n_k$ hypercube~\cite{Berger,Kim3D}, which is a $1$-CFF$(n_1+\ldots+n_k,n_1\times \ldots \times n_k)$.
Figure~\ref{array} (b) shows a $3$-dimensional hypercube, where each point represents an item and tests are given by fixing the value of one dimension. If all defective items are clustered in either a row or a column in a 2-dimensional array, we can precisely identify all of them in one round, thus this is a structure-aware $(\mathcal{ S},1)$-CFF$(2n,n^2)$ for $\mathcal{S}$ corresponding to rows and columns. 
We generalize this for higher dimensions in the next proposition.
To simplify the notation, we take $n_1=\ldots=n_k=n$, but the next results are valid for the general case.
An $[n]^k$-hypercube group testing matrix is an $1$-CFF$(kn,n^k)$ matrix defined as follows.
Items are in $\mathbb{Z}_n^k$ and rows/tests are given by $T_{v,a}=\{x\in \mathbb{Z}_n^k: x_v=a \}$, $1\leq v \leq k$, $a\in \mathbb{Z}_n$. Denote $x(v)=(x_1,\ldots,x_{v-1},x_{v+1},\ldots,x_k)$ for $x\in \mathbb{Z}_n^k$, $1\leq v\leq k$.
\vspace{-15pt}
\begin{center}
\begin{figure}[th]
\centering
\includegraphics[width=0.4\textwidth]{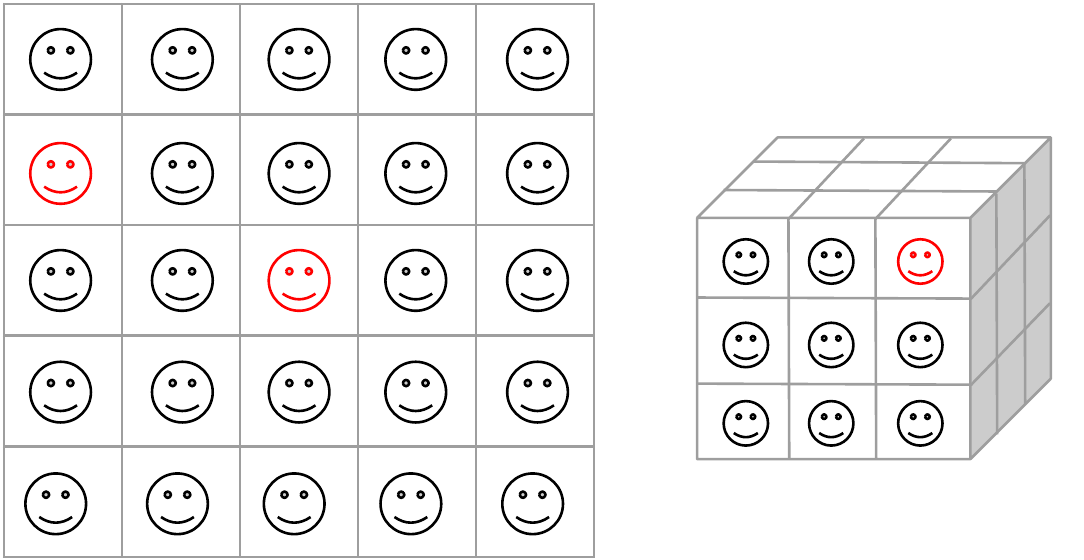}
\caption{(a) A $5\times5$ array GT with $25$ items and 10 tests. (b) A $3\times3\times3$ hypercube GT with $27$ items and $9$ tests.} \label{array}
\end{figure}
\end{center}
\vspace{-20pt}
\begin{proposition}\label{propArray}
Let $A$ be an $[n]^k$-hypercube group testing matrix.
Let $\mathcal{H}_v=([1,n],\mathcal{S}_v)$ where $\mathcal{S}_v=\{\{x\in \mathbb{Z}_n^k: x(v)=(a_1,\ldots,a_{k-1})\}:
(a_1,\cdots, a_{k-1})\in \mathbb{Z}_n^{k-1}\}$, $1\leq v \leq k$,
and let $\mathcal{H}=([1,n],\mathcal{S})$ where $\mathcal{S}=\mathcal{S}_1 \cup \cdots \cup \mathcal{S}_k$. 
Then,
 for any $1\leq v \leq k$, $A$ is an $(\mathcal{S}_v,1)$-CFF($kn,n^k$) and if $k=2$, $A$ is also an $(\mathcal{S}_v,|\mathcal{S}_v|=n)$-ECFF($2n,n^2$)$.$
Moreover, $A$ is an $(\mathcal{S},1)$-CFF($kn,n^k$).
\end{proposition}

\subsection{Construction from polynomials} 
Now we look at a construction of $d$-CFFs from polynomials over finite fields, given by Erd{\"o}s et al.~\cite{erdospoly}. 
Let $q$ be a prime power, $k$ a positive integer, and $\mathbb{F}_q = \{e_1, \ldots, e_q\}$ be a finite field. We define $\mathcal{F}=$($X, \mathcal{B}$) as follows, for each polynomial $f \in \mathbb{F}_q[x]_{\leq k}$ of degree at most $k$: $X = \mathbb{F}_q \times \mathbb{F}_q,$ %= \{(e_i,e_j): i,j = 1, \ldots, q\},$ 
 $B_f = \{(e_1, f(e_1)), \ldots, (e_q, f(e_q))\},$
 $\mathcal{B} = \{B_f: f \in \mathbb{F}_q[x]_{\leq k}\}.$
Then, $\mathcal{F}$ is a $d$-CFF($t=q^2, n=q^{k+1}$) for $d \leq \frac{q-1}{k}$.

This $d$-CFF has an interesting structure, which allows us to discard some rows when smaller values of $d$ are enough~\cite{AMC}. We restrict the CFF matrix to $i$ \emph{blocks} of rows by considering $X = \{e_1, \ldots, e_i\} \times \mathbb{F}_q$, $B_f(i) = \{(e_1, f(e_1)), \ldots, (e_i, f(e_i))\}$ and $\mathcal{B}(i) = \{B_f(i): f \in \mathbb{F}_q[x]_{\leq k}\}$, which yields the following result.
    
\begin{proposition}[Idalino and Moura~\cite{AMC}, Theorem 3.2]\label{restrictedCFF}
    Let $q$ be a prime power, $k\geq 1$ and $q \geq dk+1$, and let $\mathcal{M}$ be the $d$-CFF($q^2, q^{k+1}$) obtained from the polynomial construction. If we restrict $\mathcal{M}$ to the first $(d'k+1)$ blocks of rows, we obtain a $d'$-CFF($(d'k+1)q, q^{k+1}$), for any $d'\leq d$.
    \end{proposition}
    
For instance, for $q=5$ and $k=1$, if we restrict a $4$-CFF($5^2,5^2$) to its first $2$ blocks of rows, we get a $1$-CFF($2\times 5, q^2$), with $3$ blocks of rows we get a $2$-CFF($3\times 5, q^2$), etc. 
Next we show that this construction is an structure-aware CFF that can tolerate as many as $q$ errors with as few as $(k+1)q$ tests.

\begin{theorem}\label{theo-poly}
Let $k\geq 1$ and $q$ be a prime power such that $q\geq k+1$.
Let $\mathcal{S}= \{S_1,S_2,\ldots,S_{q^k}\}$ be a set-partition of $[1,n]$ such that $|S_i|=q$ for all $1 \leq i \leq q^k$. Then, there exists an $(\mathcal{S},1)$-CFF($(k+1)q, q^{(k+1)}$). If $k=1$, it is also an $(\mathcal{S},q)$-ECFF($2q, q^{2}$).
\end{theorem}
\begin{proof}
Each column of the 01-matrix $\mathcal{M}$ is associated with a polynomial $p\in \mathbb{F}_q[x]$ of degree at most k. Letting $\mathbb{F}_q=\{e_1, \ldots, e_q\}$, identify the blocks with $S_{i_1,\ldots,i_k}=\{ p \in \mathbb{F}_q[x] : p(e_j) = e_{i_j}, for\ 1\leq j \leq k \}$, for $(i_1,\ldots,i_k)\in [1,q]^k$.
Each row of $\mathcal{M}$ is associated with pair $(x,y)\in \mathbb{F}_q\times \mathbb{F}_q$, and $\mathcal{M}_{(x,y),p}=1$ if and only if $p(x)=y$.
Let $E\subseteq S_{i_1,\ldots,i_k}$, for some $(i_1,\ldots,i_k)\in[1,q]^k$, be a set of defective items.  
We need to show that for any column $p \notin E$, there exists a row $(c,d)$ s.t. $p(c) = d$ and $f(c) \neq d$, $\forall f \in E$.
Let $p\notin E$ and $p\in S_{j_1,\ldots,j_k}$. We consider two cases.\\
Case i) $p \in S_{i_1,\ldots,i_k} \setminus E$:
Taking $(c,d)=(e_{k+1},p(e_{k+1}))$, we know for any $f\in E$, $f(e_{k+1})\neq p(e_{k+1})$; for otherwise, since they already have the same evaluation for $e_{1}, \ldots, e_{k}$, this would imply they would be the same polynomial.\\
Case ii) $p \not\in S_{i_1,\ldots,i_k}$:
Let $\ell=min\{t: j_t\neq i_t\}$ and take $(c,d)=(e_{\ell},p(e_{\ell}))$.
We claim $\mathcal{M}_{(e_{j_{\ell}},p(e_{j_{\ell}})),p}=1$ and $\mathcal{M}_{(e_{j_{\ell}},p(e_{j_{\ell}})),f}=0$ for all $f \in E$.
Indeed, by the block definitions, for $f\in E$, $f(e_\ell) = e_{i_{\ell}} \neq e_{j_{\ell}}=p(e_{\ell})$.

If $k=1$, one block of rows in $\mathcal{M}$ has each test coinciding with each edge. Thus, $\mathcal{M}$ is also a $(\mathcal{S},q)$-ECFF($2q, q^2$).\qed

\end{proof}

As an example, for $q = 5$ and $k=1$ we have edges $S_{1} = \{0, x, 2x, 3x, 4x\}, S_{2} = \{1, x+1, 2x+1, 3x+1, 4x+1\}, S_{3} = \{2, x+2, 2x+2, 3x+2, 4x+2\}, S_{4} = \{3, x+3, 2x+3, 3x+3, 4x+3\}, \text{ and } S_{5} = \{4, x+4, 2x+4, 3x+4, 4x+4\}.$ This gives us an $(\mathcal{S},1)$-CFF($2q = 10, q^{2} = 25$) with $\mathcal{S} = \{S_1, S_2, S_3, S_4, S_5\}$, which allows us to find as many as $q=5$ defective items, as long as they are all in one of the edges $S_i$, and to find which edges are defective in any case.

Note that the construction in Theorem~\ref{theo-poly} is equivalent to a $[q]^{k+1}$-hypercube, but it is more flexible since we can add more tests (Proposition~\ref{restrictedCFF}) for a total of $(dk+1)q$ tests, where $q\geq dk+1$, to obtain both a $(\mathcal{S},1)$-CFF$((dk+1)q,q^{k+1})$ and a $d$-CFF$((dk+1)q,q^{k+1})$, so any $d$ defects anywhere  or $q>d$ defects inside an edge can be found.

\section{Structure-aware CFFs: overlapping edges}\label{sec:overlap}

Here, edge colouring of hypergraphs is used to partition the edges of the graph into sets of non-overlapping edges (colour classes) allowing the use of previous constructions to deal with each colour class.
An $\ell$-edge-colouring of a hypergraph $\mathcal{H}=(V,\mathcal{S})$ is a mapping from $\mathcal{S}$ to $\{1,\ldots, \ell\}$ such that no vertex is incident to more than one edge mapping to the same colour.
Let $\chi'(\mathcal{H})$ be the edge chromatic number of hypergraph $\mathcal{H}$, which is the minimum $\ell$ among all $\ell$-edge-colourings.

\begin{theorem}~\label{hyperpacking}
Let $\mathcal{H}=([1,n],{\mathcal S})$ be a hypergraph and let ${\mathcal C}_1, \mathcal{C}_2, \ldots, \mathcal{C}_\ell$ be the sets of edges in each colour class of an $\ell$-edge-colouring of $\mathcal{H}$. For each $i$, $1\leq i \leq \ell$, let $k_i=\max\{|A|:A\in \mathcal{C}_i\}$ and let $f_i=|\mathcal{C}_i|+\delta_i$ , where $\delta_i=0$ if the $\mathcal{C}_i$ spans $[1,n]$ and $\delta_i=1$, otherwise. Then, given $1$-CFF$(t_i,f_i)$ for $1\leq i \leq \ell$, we can construct a $(\mathcal{S},1)$-CFF$(t,n)$  where
$t=\sum_{i=1}^\ell (t_i + k_i) $; moreover we can construct
an $(\mathcal{S},1)$-ECFF$(t_i,n)$.
\end{theorem}
\begin{proof} We first take each colour class plus possibly an additional set so that this set of edges span $[1,n]$.
Let $E_i=[1,n]\setminus(\cup_{S \in \mathcal{C}_i} S)$, $1\leq i \leq \ell$.
Let $\mathcal{F}_i=\mathcal{C}_i$ if  $E_i=\emptyset$ and 
$\mathcal{F}_i=\mathcal{C}_i \cup \{E_i\}$, otherwise.
Let $\mathcal{H}_i=([1,n],\mathcal{F}_i)$ and note $f_i=|\mathcal{F}_i|$. Since $\mathcal{H}_i$ is a hypergraph with non-overlapping edges, we apply a construction inspired
by Proposition~\ref{prop:kroneckerplus} to build a structure-aware CFF for $\mathcal{H}_i$ with $r=1$. Indeed, if each edge has the same cardinality $k_i$ and the colour class spans $[1,n]$, letting $c_i=n/k_i$, we just use a $1$-CFF$(t_i,c_i)$  for $A$ and $B=R_{c_i}$, a row of all 1's which is a $0$-CFF, and apply Proposition~\ref{prop:kroneckerplus}. This means we vertically concatenate $A\otimes R_{k_i}$ with $R_{c_i} \otimes I_{k_i}$ to get a  $({\cal C}_i,1)$-CFF$(t_i+k_i,n)$. We describe next the general case.

Let $\mathcal{F}_i=\{F_1, \ldots, F_{f_i}\}$ and recall that $\mathcal{F}_i$ is a partition of $[1,n]$. 
Let $A_i$ be a $1$-CFF$(t_i,f_i)$.
Let $M_i$ be a $t_i\times n$ build from $A_i$
where column $c$ of $M_i$ is obtained from column $j$ of $A_i$
if vertex $v_c$ is in $F_j$.
Let $N_i$ be a $k_i\times n$ array consisting of a kind of ``identity matrix'' under each edge of  $\mathcal{C}_i$.
More precisely, for each edge $F_j=\{v_1,\ldots,v_{|F_j|}\} \in \mathcal{F}$, $F_j\in\mathcal{C}_i$ the column corresponding to $v_x$ has a 1 in row $x$ and zero elsewhere, $1\leq x \leq |F_j|$. If $E_i\not=\emptyset$, place a column of 0's under vertices in $E_i$.

Now, we vertically concatenate all arrays $M_1, M_2, \ldots, M_\ell, N_1, N_2, \ldots, N_\ell$ to form a $t\times n$ array $\mathcal{M}$.
Next we show that $\mathcal{M}$ is a $(\mathcal{S},1)$-CFF$(t,n)$.
Let $S\in \mathcal{S}$. Take $I\subseteq S$ and $i_0\in[1,n]\setminus I$. $S$ must be in some colour class $\mathcal{C}_i$.
For the case $i_0\in S$, the sub-array formed by columns of $N_i$ indexed by $I\cup \{i_0\}$ contains a row $w$ where $N_i[w,i_0]=1$ and $N_i[w,j]=0$ for all $j\in I$.
For the case $i_0\not\in S$, the sub-array formed by columns of $M_i$ indexed by $I\cup \{i_0\}$ contains a row $w$ such that 
$M_i[w,i_0]=1$ and $M_i[w,j]=0$ for all $j\in I$. This is because $M_i$ is a $1$-CFF on the edges of $F_i$ and $I$ and $i_0$ are each contained in two separate edges of $\mathcal{H}_i$, since
$I\subseteq S \in \mathcal{F}_i$ and $i_0\notin S$ and $\mathcal{F}_i$ spans $[1,n]$.
Thus in both cases, the condition $|B_{i_0}\setminus (\cup_{i \in I} B_i)| \geq 1$ in Equation~\ref{eqn:vcff} of Definition~\ref{variabledef} is satisfied. Thus $\mathcal{M}$ is an  $(\mathcal{S},1)$-CFF$(t=t_1+k_1+\ldots+t_\ell +k_\ell,n)$. It is easy to see that if we only use matrixes $M_1, \ldots, M_\ell$ we obtain an $(\mathcal{S},1)$-ECFF$(t=t_1+\ldots+t_\ell ,n)$.

 \qed
\end{proof}

\begin{corollary} \label{hyperpackinguniform}
Let $\mathcal{H}=([1,n],\mathcal{S})$ be a $k$-uniform hypergraph. Denote by $t(1,x)$ the number of tests $y$ in the Sperner construction of a $1$-CFF$(y,x)$ (See Section~\ref{sec:sperner}).
Then, there exists an $(\mathcal{S},1)$-CFF$(t,n)$ with
$t = \chi'(\mathcal{H})\times (t(1,\lceil (n/k)\rceil)+k) \sim \chi'(\mathcal{H}) ((\log n/k) + k) $, and a $(\mathcal{S},1)$-ECFF$(t',n)$ with 
$t' =  \chi'(\mathcal{H})\times t(1,\lceil (n/k)\rceil\sim \chi'(\mathcal{H}) \log n/k$.
\end{corollary}
\begin{proof}
We can apply Theorem~\ref{hyperpacking} with $\ell=\chi'(\mathcal{H} )$, $k_i=k$ for all $1\leq i\leq \ell$ and note that each colour class contains at most $\lfloor n/d \rfloor$ edges so that $f_i \leq \lceil n/d \rceil$. Note that if each colour class spans $[1,n]$ this is equivalent to applying Proposition~\ref{prop:spernerunlimited} to the non-overlapping hypergraph given by each colour class.\qed
\end{proof}

\begin{example}\label{ex:highschool}
Consider a high school where each student takes $P$ courses per term, in $P$ weekly time periods where in each time period each student attends one courses of their choice. Consider a hypergraph with $n$ vertices corresponding to students and each edge corresponding to students in a course.
In this example $\ell=P$, since each time period forms a colour class. 
We give a tiny example, with $n=18$ students spread of over $P = 2$ time periods morning/afternoon each with $6$ optional courses with $3$ students each. This hypergraph has $m=12$ edges and $n=18$ vertices displayed in the table below. \\

\hspace{-8mm}
{\tiny
\begin{tabular}{c|c|c|c|c|c|c|c|c|c|c|c|c|c|c|c|c|c|c|}
 students:   &  01 & 02 & 03 & 04 & 05 & 06
             &  07 & 08 &09 &10 & 11 &12
             &  13 & 14 &15 &16 & 17 &18\\ \hline 
 course 1& X&X&X& & & & & & & & &  & & & & & & \\
 course 2&  & & &X&X&X& & & & & & & & & & & & \\     
 course 3&  & & & & & &X&X&X& & & & & & & & & \\ 
 course 4&  & & & & & & & & &X&X&X& & & & & &\\ 
 course 5&  & & & & & & & & & & & &X&X&X& & & \\ 
 course 6&  & & & & & & & & & & & & & & &X&X&X\\ \hline
 course 7& X&X& &X& & & & & & & & & & & & & & \\
 course 8&  & &X& &X& & &X& & & & & & & & & & \\     
 course 9&  & & & & &X&X& &X& & & & & & & & & \\ 
 course 10& & & & & & & & & &X& &X& & & & & &X \\ 
 course 11& & & & & & & & & & & & & &X& &X&X& \\ 
 course 12& & & & & & & & & & & & &X& &X& & &X\\  \hline 
\end{tabular}
}
\
{
\tiny
\begin{tabular}{c}
\begin{tabular}{ l|c|c|c|c|c|c|} \cline{2-7}
test \ 1:& 111&111&111&000&000&000\\
test \ 2:& 111&000&000&111&111&000\\
test \ 3:& 000&111&000&111&000&111\\
test \ 4:& 000&000&111&000&111&111\\
\cline{2-7} 
test \ 5:& 100&100&100&100&100&100\\
test \ 6:& 010&010&010&010&010&010\\
test \ 7:& 001&001&001&001&001&001\\ \hline
\end{tabular}\\

\begin{tabular}{l|p{2.6cm}|} \cline{2-2}
test \ 8:& \\
test \ 9:&  \ \ \ permute columns\\
test 10:& \ \ \ of above array\\
test 11:& \ \ \   so that blocks\\
%\cline{2-2}
test 12:& \ \ \ of 3 columns are \\
test 13:& \ \ \ placed under edges \\
test 14:& \ \ \ of second period\\
\cline{2-2}
\end{tabular}
\end{tabular}
}
\end{example}
\begin{example}\label{ex:highschool2}
Consider the setup of Example~\ref{ex:highschool}. Let us consider a more realistic scenario of a high school with students taking 4 courses each term like the ones in Ontario, Canada. Suppose $n=900$ students take $P=4$ courses each, each course having $30$ students for a total of $m=120$ courses.
The matrix for each time period $i$ can be build from a 1-CFF$(7,30=120/4)$ to form $M_i$ and identities of order $30$ side-by-side to form $N_i$. Assume there is an outbreak in a single course, involving any number of students ($\leq 30$) in that course. 
We only need $7\times4 = 28$ tests to determine the course where the outbreak took place ($M_i$ build from 1-CFF(7,30), $1\leq i \leq 4$).
A total of $28 + 30\times4 = 148$ tests can be used to identify all infected individuals (up to 30) in this set of $900$ students. Note that our assumption is that there is $r$=1 course that contains all infected individuals, even thought there may be many infected courses (say up to 90 other courses that the infected students also take in other time periods).
In other words the hypergraph is assumed to have a defective cover of size $r=1$, but it is possible that up to $91$ edges are defective.
\end{example}

For $r>1$, we need to use strong edge-colourings to be able to split the problem according to colour classes without too many infected edges appearing in the same colour class.
A strong edge-coloring of a hypergraph $\mathcal{H}$ is an edge-coloring such that any two vertices belonging to distinct edges with the same colour are not adjacent. The strong chromatic index $s'(\mathcal{H})$ is the minimum number of colors in a strong edge-coloring of $\mathcal{H}$.

\begin{theorem}~\label{hyperpacking2}
Let $\mathcal{H}=([1,n],\mathcal{S})$ be a hypergraph and let $r\geq 2$ be an upper bound on the number of edges of a minimal defective cover.
Let $\mathcal{C}_1, \mathcal{C}_2, \ldots, \mathcal{C}_{s'}$ be the sets of edges in each colour class of an $s'$-strong-edge-colouring of $\mathcal{H}$. Let $k_i=\max\{|S|:S\in \mathcal{C}_i\}$.
Then there exists an $(\mathcal{S},r)$-CFF$(t,n)$  with
$t \leq \sum_{i=1}^{s'} (t(r,|\mathcal{C}_i|) +  k_i t(r-1,|\mathcal{C}_i|)).$
\end{theorem}

\begin{proof} 
We can assume w.l.o.g. that every vertex belongs to some edge, otherwise, the vertex can be eliminated from the problem, since it cannot be defective.
The construction is similar to the one in Theorem~\ref{hyperpacking}, but we do not add the extra dummy edge to a colour class.
For each colour class $\mathcal{C}_i = \{C_1, C_2, \ldots, C_{|\mathcal{C}_i|}\}$, we build a matrix $M_i$ and $N_i$. 
For any set of edges $\mathcal{E}=\{S_1, \ldots, S_r\}\subseteq \mathcal{S}$, there is at most $r$ edges in $\mathcal{C}_i$ which intersect any edge in $\mathcal{E}$, due to the definition of strong colouring.
Let $t_i=t(r,|\mathcal{C}_i|).$
Build a $t_i\times n$ matrix $M_i$ using an $r$-CFF$(t_i,|\mathcal{C}_i|)$, say $A_i$, 
where column $c$ of $M_i$ is obtained from column $j$ of $A_i$ if vertex $v_c$ is in edge $C_j\in \mathcal{C}_i$, or is a zero column if $v_c$ is not contained in any edge of $\mathcal{C}_i$.
Since at most $r$ edges are defective in colour class $\mathcal{C}_i$, $M_i$ is enough to determine which edges  in $\mathcal{C}_i$ are defective.
Let $t'_i=t(r-1,|\mathcal{C}_i|).$
To build matrix $N_i$, we use  an $(r-1)$-CFF$(t'_i,|\mathcal{C}_i|)$, say $B_i$, in a construction that is similar to a Kronecker product with an identity matrix $I_{k_i}$, but not quite, since not all edges have cardinality $k_i$.
More precisely, for each edge $C_j=\{v_{i_1}, \ldots, v_{i_{|C_j|}} \}\in \mathcal{C}_{i}$ use column $j$ of  $B_i$ corresponding to $C_j$ to build a $(t'_i\cdot k_i) \times |C_j|$ matrix for this edge as follows. The element at position $B_i[x,j]$ is substituted by a $k_i\times |C_j|$ matrix that consists of the first $|C_j|$ columns of an identity matrix of dimension $k_i$ if $B_i[x,j]=1$, or a 0-matrix if $B_i[x,j]=0$.
The columns of this $(t'_i\cdot k_i)\times |C_j|$ matrix are pasted as columns of $N_i$ that correspond to vertices $v_{i_1}, \ldots, v_{i_{|C_j|}}$, respectively. 
For vertices not contained in any edge in colour class $\mathcal{C}_i$, paste a column of zeroes. This matrix is enough to determine which items in the defective edges of this colour class are defective.
Vertically concatenate $M_1, M_2, \ldots, M_{s'}, N_1, N_2, \ldots, N_{s'}$.
We must show that this is an $(\mathcal{S},r)$-CFF($\sum_{i=1}^{s'}(t_i+k_i t'_i),n$).
Let $I\subseteq \cup_{j=1}^{\ell} S_j$, $\ell\leq r$ and
let $i_0\in [1,n]\setminus I$.
If there exists an edge $S$ such that $i_0 \in S$ and $S\cap I=\emptyset$, we claim that in $M_i$ corresponding to the colour class $\mathcal{C}_i\supset\{S\}$ , we can find 
a row $x$ where $M_i[x,i_0]=1$ and $M_i[x,y]=0$ for all $y\in I$.
This is because in this colour class there are at most $\ell\leq r $ edges that intersect $I$  (due to strong colouring and the fact that $I\subseteq \cup_{j=1}^{\ell} S_j$), and since $M_i$ was built from an $r$-CFF on the edges of this colour class, in some row $x$ columns corresponding to edge $S$ contains 1's while columns corresponding to the other edges intersecting $I$ must have zeros.
The other case to analyse is when $i_0\in S$ such that $S\cap I\not=\emptyset$.
Then $S \cap S_{j_S} \not=\emptyset$ for some $1\leq j_S \leq \ell$. Let $\mathcal{C}_v$ be the colour class that contains $S$. In $N_v$ we claim that we can find a row $x$ where $N_v[x,i_0]=1$ and $N_v[x,y]=0$ for all $y \in I$.
Indeed, due to strong colouring, there are at most $\ell-1\leq r-1$ edges that pass through some vertex in $(\cup_{j=1}^{\ell} S_j)\setminus S_{j_S}$
and $S_{j_S}$ cannot be in the same colour class as $S$.
Since $B_v$ is an $(r-1)$-CFF on the edges of  colour class $\mathcal{C}_v$, there exists
a row $z$ with $B_v[z,S]=1$ and $B_v[z,S_j]=0$ for all $1\leq j\leq \ell$, $j\not= j_S$.
When row $z$ is substituted by the rows corresponding to an identity matrix under $S$, we find a row $x$ of $N_v$ such that $N_v[x,i_0]=1$ and $N_v[x,p]=0$ for all $p \in S\setminus\{i_0\}$ (which implies $N_v[x,p]=0$ for all $p \in S_{j_S}\setminus\{i_0\}$)
and also $N_v[x,q]=0$, for all $q \in S_j$, $1\leq j\leq \ell$, $j\not= j_S$.
This concludes the proof. 
\qed
\end{proof}
\begin{figure}
    \centering
    \includegraphics[width=8cm]{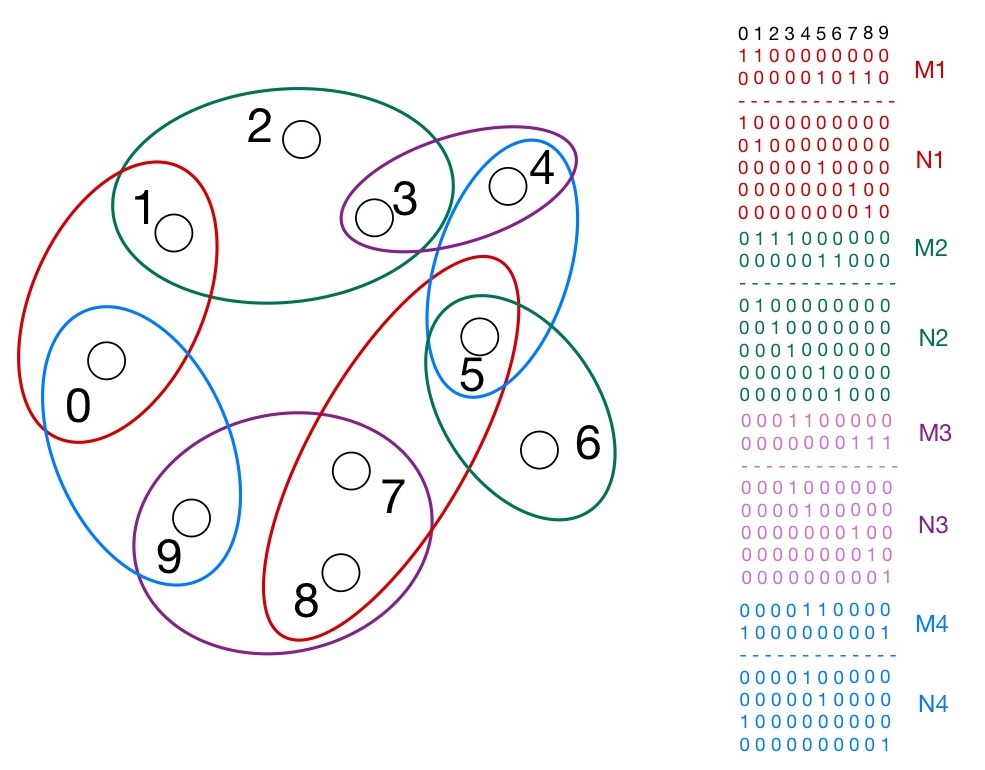}
    \caption{Example application of Theorem~\ref{hyperpacking2}. Due to small size of colour classes, the 1-CFF and 2-CFF used as ingredients are identity matrices $I_2$.}
    \label{fig:my_label}
\end{figure}

\begin{corollary} \label{hyperpackinguniform2}
Let $\mathcal{H}=([1,n],\mathcal{S})$ be a $k$-uniform hypergraph and let $\Delta$ be the maximum degree of a vertex.
Then, we can build a $(\mathcal{S},r)$-CFF$(t,n)$ with $t \leq s'(H)\times (t(r,\lfloor n/k\rfloor)+k t(r-1,\lfloor n/k\rfloor)) \leq (k \Delta +1)\times(t(r,\lfloor n/k\rfloor)+k t(r-1,\lfloor n/k\rfloor)).$
\end{corollary}
\begin{proof}
We can apply Theorem~\ref{hyperpacking2} with $s'=s'(H)$, $k_i=k$ for all $1\leq i\leq \ell$ and note that each colour class contains at most $\lfloor n/k \rfloor$ edges so that $|\mathcal{C}_i| \leq \lfloor n/k \rfloor$. In addition, $s'(\mathcal{H}) \leq (k \Delta +1)$, since a greedy colouring algorithm using $(k \Delta +1)$ colours always succeeds to find a strong colouring. \qed

\end{proof}

\begin{example}
Consider the scenario of Example~\ref{ex:highschool} and $r=2$.
We can find a strong colouring for the hypergraph of that example with $s'=6$ colours with colour classes:
%\begin{center}
$\{course 1, course 4\}, \{course 2, course 5\}, \{course 3,course 6\},\\ \{ course 7,course 10\},
\{course 8,course 11 \}, \{course 9,course 12\}.$
%\end{center}
For each colour class we can use identity matrices $I_2$ as the $2$-CFF$(2,2)$ and $I_2$ as the $1$-CFF$(2,2)$ required so that $t_i+k_it'_i=8$.
If there are outbreaks in 2 courses, any set of up to $6$ students in these 2 courses can be detected with $48$ tests.
This is a toy example, and of course using $6\times8 = 48$ tests is not worth it, since it is better testing the $18$ students individually.
The construction would be advantageous if we have less colour classes with more edges in each, like in the next example.

\end{example}

\begin{example}\label{ex:grid}
Consider a venue with $4356$ people sitting in a square auditorium of 66 rows with 66 seats per row.
Edges are sets of individuals sitting nearby.
We consider edges of size $9$ consisting of all possible contiguous $3 \times 3$ squares (see Fig.~\ref{EX5grid}); there are 9 edges passing through each vertex as shown in Fig.~\ref{EX5edges}.
There is a strong colouring with $\ell=36$ colour classes of $11\times 11= 121$ edges each: we need 4 colours to ``tile" the room with edges and 9 such tilings to cover all edges (see Fig.~\ref{EX5strong}).
For each colour class we use a $2$-CFF($25,125$) using the polynomial construction from Proposition~\ref{restrictedCFF} for $q=5$ and $k=2$. This gives the part corresponding to the matrices $M_1, \ldots, M_{36}$ in Theorem~\ref{hyperpacking2} totalling $36\times 25=900$ tests.
For the $N_1, \ldots N_{36}$, each of which is supposed to be a 1-CFF($t,121$) multiplied by $I_9$, we use instead a single matrix $N$ built as follows. 
Take $A$ as a $1$-CFF($12,484$) obtained from the Sperner construction and do $N=A\otimes I_9$ with $108$ rows. Carefully assign vertices in the grid to the columns of matrix $N$ so that each $3\times 3$ square corresponds to a block of identity matrix $I_9$ in a tiling fashion (see Fig.~\ref{EX5_N}). This is enough to identify each non-defective vertex that lies inside one of the two defective edges, which is the purpose of $N$. Therefore with a total of $1008$ tests we can screen $4356$ people for any $18$ infected people that appear within any $2$ regions of size $3\times 3$.
\begin{figure}[h]
\centering
\includegraphics[width=0.3\textwidth]{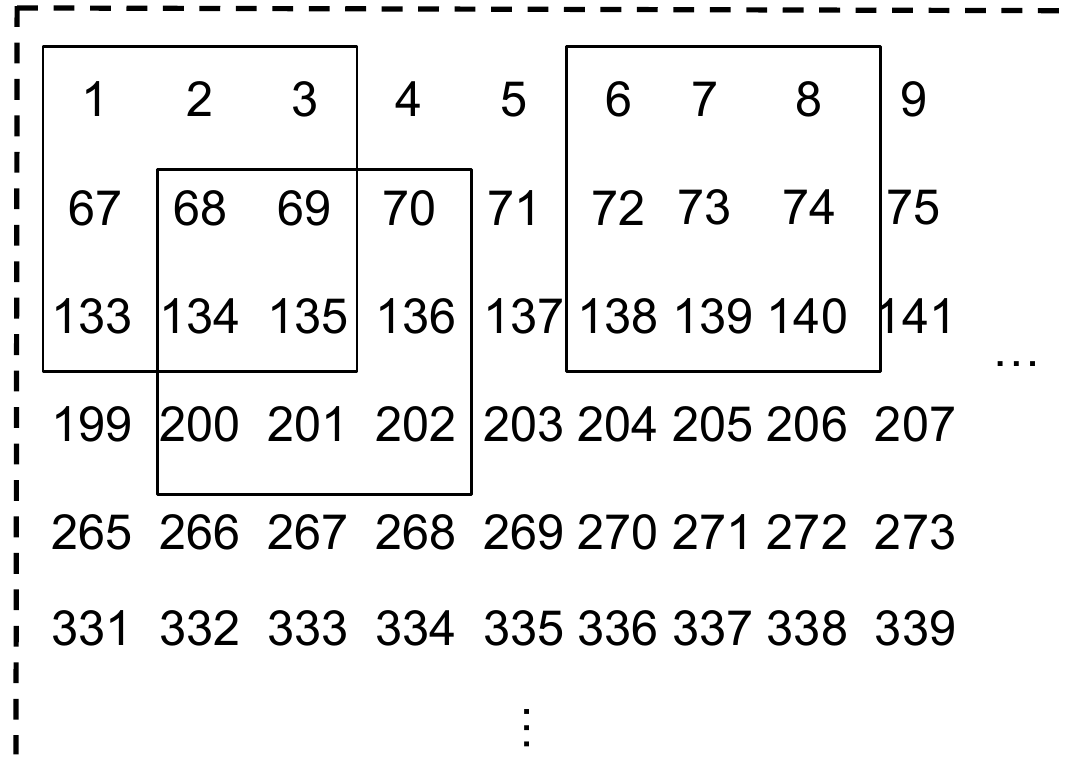}
\caption{Each vertex is a person in a $66\times66$ auditorium numbered $1$ to $4356$. Edges are sets of 9 people sitting nearby ($3\times3$ squares); 3 different edges are shown. There is less than $4356$ edges, since edges centered at extreme vertices are not used; for example no edge has vertex 1, 4 or 199 as their center.
More precisely, there are $4094$ edges.} \label{EX5grid}
\end{figure}

\begin{figure}[h]
\centering
\includegraphics[width=0.3\textwidth]{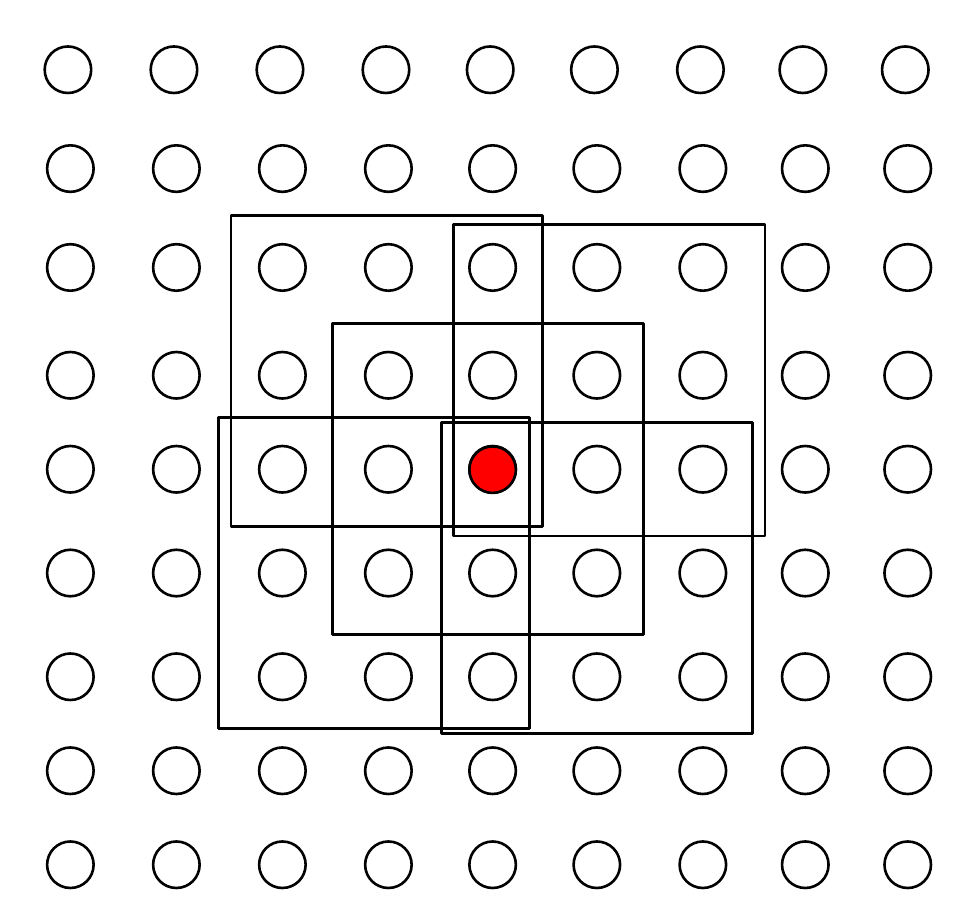}
\includegraphics[width=0.3\textwidth]{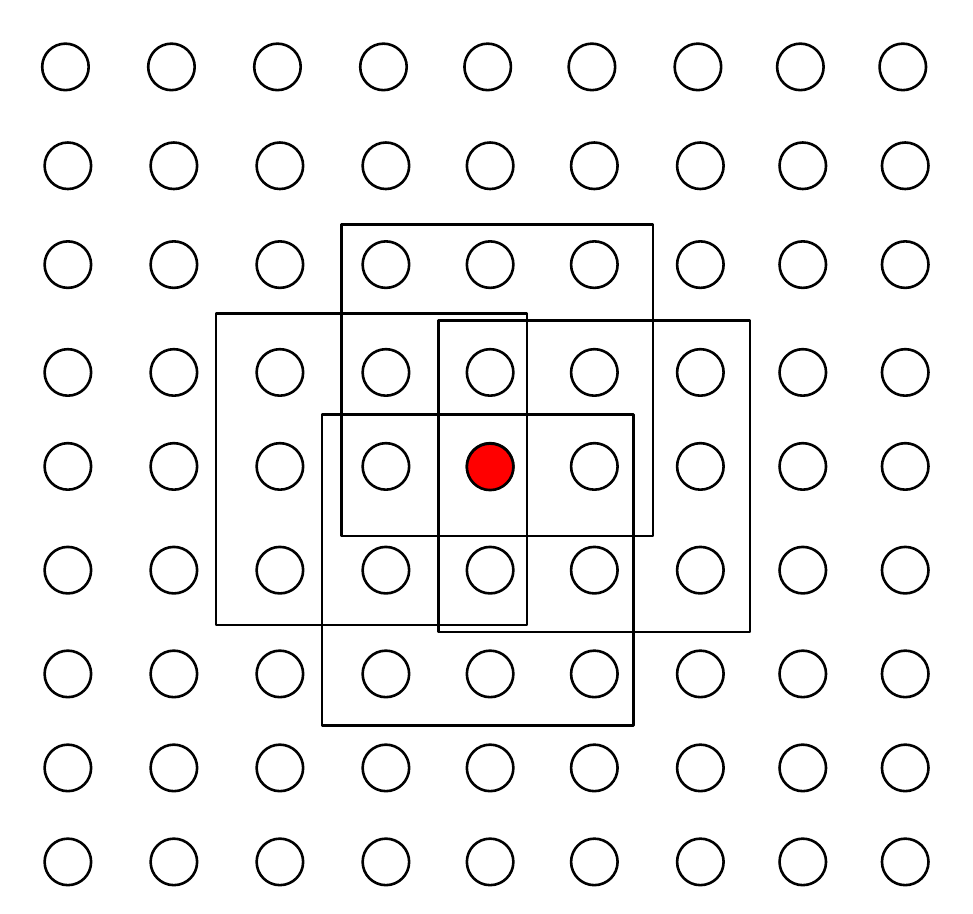}
\caption{Nine edges passing through a vertex. Right picture showing $5$ of them, left one showing the remaining $4$. The squares are the regions that each person belongs to.} \label{EX5edges}
\end{figure}

\begin{figure}[h]
\centering
\includegraphics[width=0.3\textwidth]{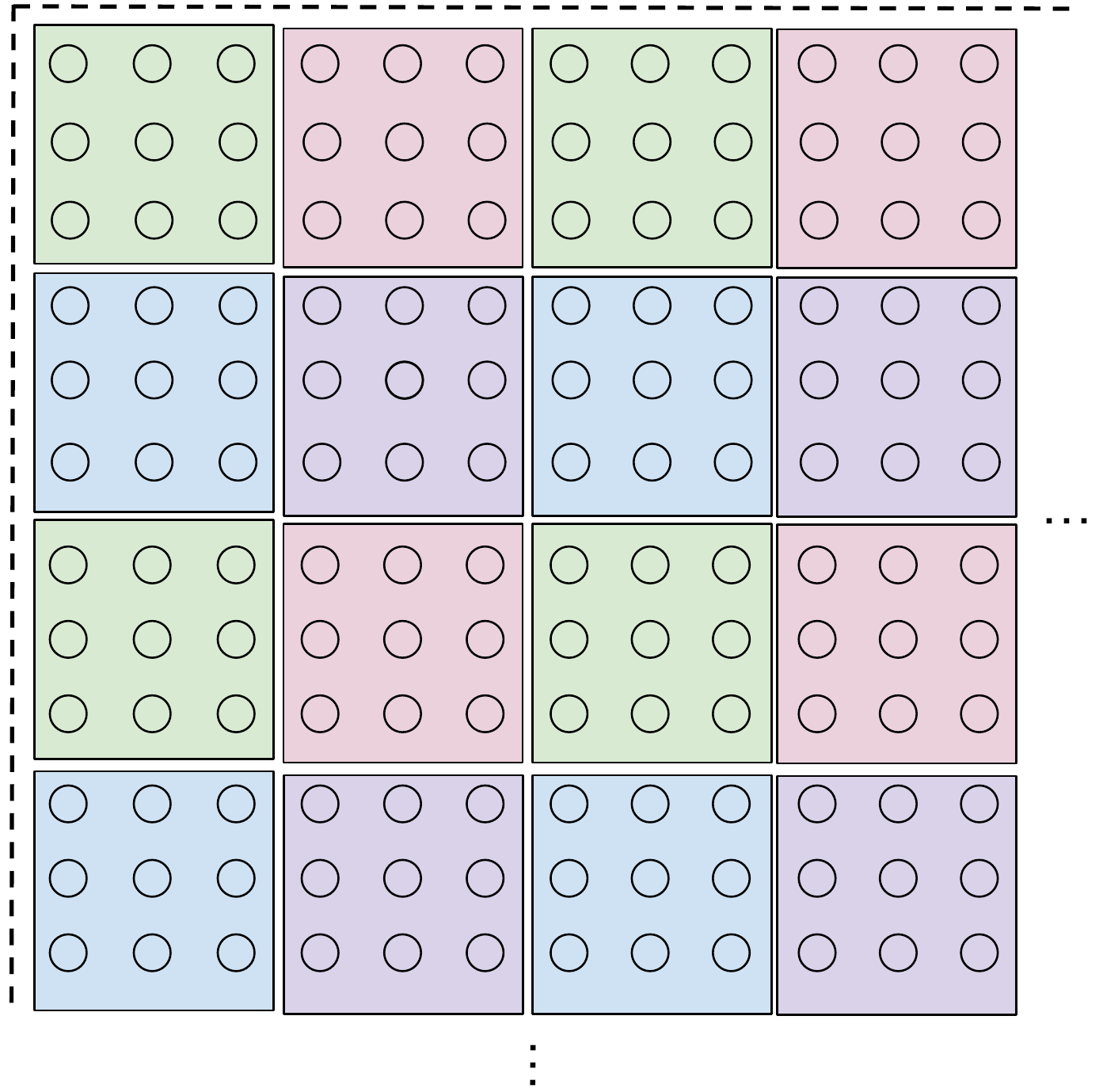}
\includegraphics[width=0.3\textwidth]{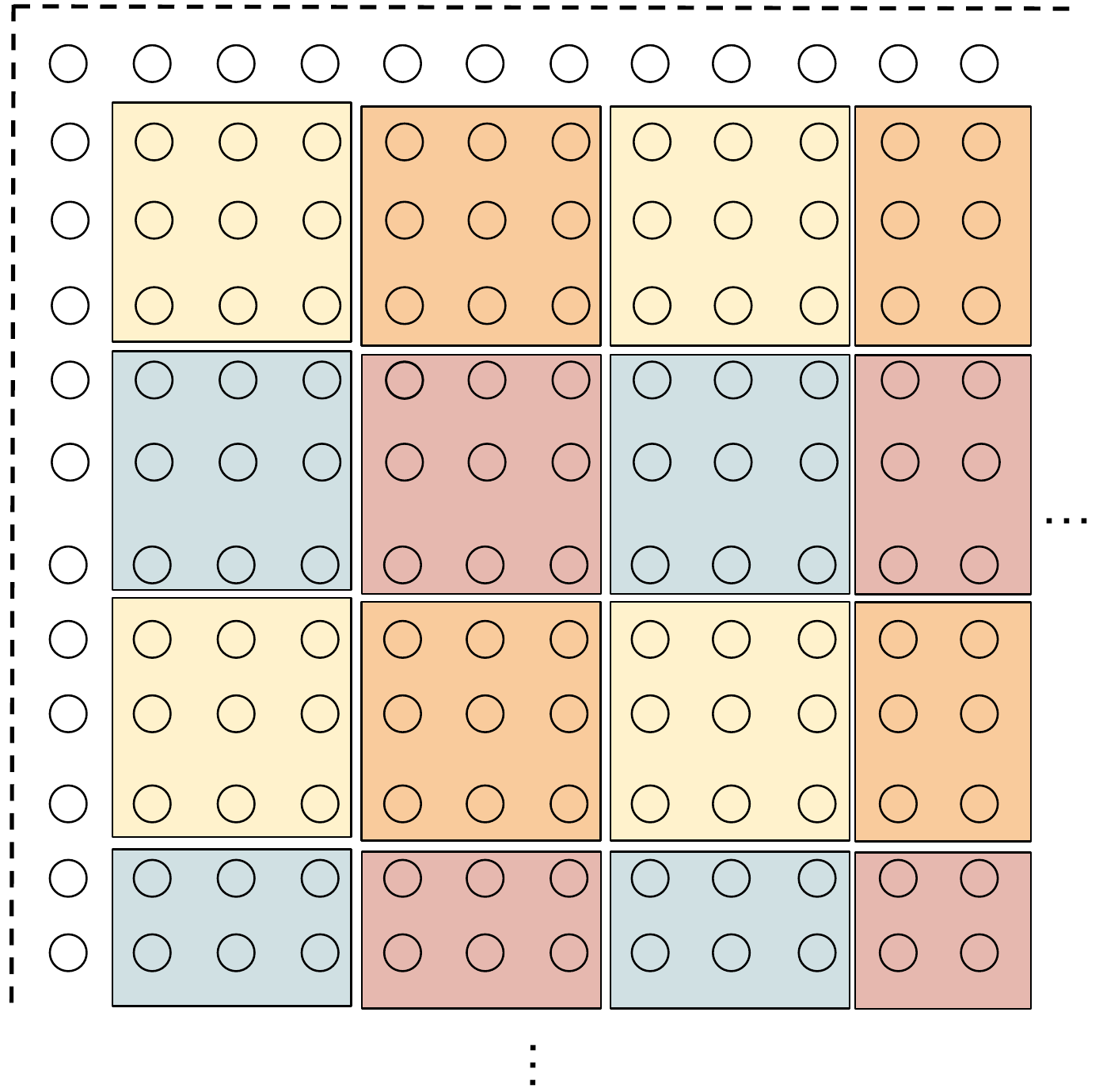}
\caption{A strong colouring requires 36 colour classes. We can cover all the edges with 9 possible tilings of the grid with 4 colours in each tiling, yielding 36 colour classes. Only two of the tilings are shown.} \label{EX5strong}
\end{figure}

\begin{figure}[h]
\centering
\includegraphics[width=0.3\textwidth]{Desenho3_a.pdf} 
\includegraphics[width=0.3\textwidth]{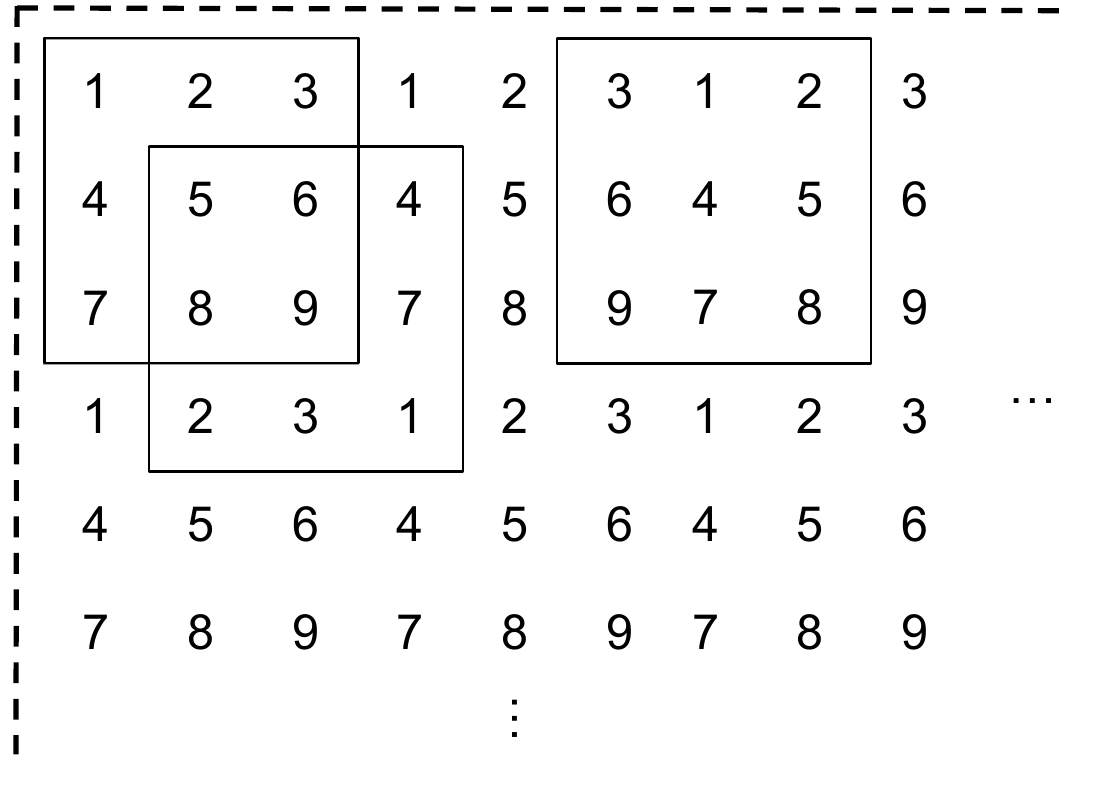}
\caption{Each vertex in the left picture is assigned one of the 9 columns of the identity matrix $I_9$ according to the right picture. This allows to create matrix $N$ which is the cross product of a $1$-CFF$(12,484)$ array construction and the identity matrix $I_9$. }
 \label{EX5_N}
\end{figure}
\end{example}


\newpage
\begin{thebibliography}{8}


\bibitem{Berger}
Berger, T., Mandell, J. W.,  Subrahmanya, P. Maximally Efficient Two-Stage Screening. Biometrics, \textbf{56}(3), 833--840, (2000).

\bibitem{Bshouty}
Bshouty, N. H., Linear Time Constructions of Some Restriction Problems. In: Paschos V., Widmayer P. (eds) Algorithms and Complexity. Lecture Notes in Computer Science, vol 9079. Springer, Cham (2015).


\bibitem{Dorfman}
Dorfman, R. The detection of defective members of large population. The Annals of Mathematical Statistics, 14, 436--440 (1943).


\bibitem{DuHwang}
Du, D-Z., Hwang, F. K. Combinatorial group testing and its applications. 2nd edn.World Scientific, Singapore (2000).


\bibitem{TNYT}
Ellenberg, J. Five People. One Test. This Is How You Get There. The New York Times, \url{https://www.nytimes.com/2020/05/07/opinion/coronavirus-group-testing.html}. Last accessed 11 Jan 2022.

\bibitem{eps2007}
Eppstein, D., Goodrich, M. T., Hirschberg, D. S. Improved Combinatorial Group Testing Algorithms for Real‐World Problem Sizes. SIAM Journal on Computing, \textbf{36}(5), 1360--1375 (2007).


\bibitem{erdospoly}
Erd{\"o}s, P., Frankl, P., F{\"u}redi, Z.
Families of finite sets in which no set is covered by the union of r others.
Israel Journal of Mathematics, \textbf{51}, 79 -- 89 (1985).

\bibitem{far1997}
Farach-Colton, M., Kannan, S., Knill, E., Muthukrishnan, S. Group testing problems with sequences in experimental molecular biology. Proceedings. Compression and Complexity of sequences, 357--367. IEEE Press, Washington, DC (1997).

\bibitem{furedi}
F\"{u}redi, Z. On r-Cover-free Families.
Journal of Combinatorial Theory, Series A, \textbf{73}, 172--173 (1996).

\bibitem{Vaccaro2}
	Gargano, L., Rescigno, A.A., Vaccaro, U.
Low-Weight Superimposed Codes and their Applications,
	 In: Proceedings of the 12th International Frontiers of Algorithmics Workshop (FAW'18). Lectures Notes in Computer Science, vol. 10823, 197--211 (2018).

\bibitem{Vaccaro}
Gargano, L., Rescigno, A.A., Vaccaro, U. Low-weight superimposed codes and related combinatorial structures: Bounds and applications. Theoretical Computer Science, 806, 655--672 (2020).	


\bibitem{Hudgens}
Hudgens, M. G., Kim, H. Y. Optimal Configuration of a Square Array Group Testing Algorithm. Communications in statistics: theory and methods, \textbf{40}(3), 436–448 (2011).

\bibitem{IPL}
 Idalino, T. B.,  Moura, L.,  Cust{\'o}dio,R. F., Panario, D. Locating modifications in signed data for partial data integrity. Information Processing Letters, \textbf{115}, 731 -- 737 (2015).

\bibitem{IWOCA}
Idalino,T.B., Moura, L.
\newblock Efficient Unbounded Fault-Tolerant Aggregate Signatures Using Nested Cover-Free Families
\newblock In \emph{International Workshop on Combinatorial Algorithms, IWOCA 2018}, LNCS, vol. 10979, pp. 52--64. Springer, Cham. 

\bibitem{AMC}
Idalino,T.B., Moura, L. Embedding cover-free families and cryptographical applications. Advances in Mathematics of Communications, \textbf{13}(4), 629--643 (2019).

\bibitem{PhDThesis}
Idalino, T.B., Fault Tolerance in Cryptographic Applications Using Cover-Free Families. PhD Thesis, University of Ottawa, Ottawa, Canada, 2019.

\bibitem{indocrypt}
Idalino,T.B., Moura, L., Adams, C.
Modification tolerant signature schemes: location and correction. In: Hao F., Ruj S., Sen Gupta S. (eds) Progress in Cryptology – INDOCRYPT 2019. Lecture Notes in Computer Science, vol 11898, pp. 23--44. Springer, Cham (2019).

\bibitem{TCS}
Idalino,T.B., Moura, L.
\newblock Nested Cover-Free Families for Unbounded Fault-Tolerant Aggregate Signatures.
\newblock Theoretical Computer Science, \textbf{854}, 116--130 (2021).

\bibitem{Kautz}
Kautz, W., Singleton, R. Nonrandom binary superimposed codes. IEEE Transactions on Information Theory, \textbf{10}, 363--377 (1964).

\bibitem{Kim}
Kim, H. Y., Hudgens, M.G., Dreyfuss, J.M., Westreich, D.J. and Pilcher, C.D. Comparison of Group Testing Algorithms for Case Identification in the Presence of Test Error. Biometrics, 63, 1152--1163 (2007).

\bibitem{Kim3D}
Kim, H. Y., Hudgens, M. G. Three-dimensional array-based group testing algorithms. Biometrics, \textbf{65}(3), 903--910 (2009).

\bibitem{monotone} 
Li, P.C., van Rees, G.H.J., Wei, R.
Constructions of 2-cover-free families and related separating hash families,
Journal of Combinatorial Designs, \textbf{14}, 423--440 (2006).
	
\bibitem{nature}
Mallapaty, S.: The mathematical strategy that could transform coronavirus testing. Nature, \textbf{583}(7817), 504--505 (2020).

\bibitem{Nikolopoulos1}
Nikolopoulos, P., Rajan Srinivasavaradhan, S., Guo, T., Fragouli, C., Diggavi, S. Group testing for connected communities. In: 24th International Conference on Artificial Intelligence and Statistics
on Proceedings of Machine Learning Research, pp. 2341--2349. PMLR, (2021).

\bibitem{Nikolopoulos2}
Nikolopoulos, P., Rajan Srinivasavaradhan, S., Guo, T., Fragouli, C., Diggavi, S. Group testing for overlapping communities. In: ICC 2021 - IEEE International Conference on Communications, pp. 1--7. IEEE (2021).

\bibitem{Sudbury}
Phatarfod, R. M. and Sudbury, A. The use of a square array scheme in blood testing. Statistics in Medicine, \textbf{13}(22), 2337--2343 (1994).

\bibitem{PR}
Porat, E., Rothschild, A. Explicit nonadaptive combinatorial group testing schemes. IEEE Transactions on Information Theory, \textbf{57}, 7982--7989 (2011).

\bibitem{Ruszinko}
	Ruszink\'{o}, M. On the upper bound of the size of the r-cover-free families. Journal of Combinatorial Theory, Series A, \textbf{66}, 302--310 (1994).

\bibitem{Verdun}
Verdun et al. Group Testing for SARS-CoV-2 Allows for Up to 10-Fold Efficiency Increase Across Realistic Scenarios and Testing Strategies. Frontiers in Public Health, 9, (2021).

\bibitem{wei2006}
Wei, R. On cover-free families. Technical report, Lakehead University, 2006.
\vspace{160pt}

\end{thebibliography}
\end{document}